\def\llncs{1}       
\def\setmargin{0}   
\def\papertype{1}   
\def\anonymous{0}   
\def\fullversion{1} 

\documentclass[twocolumn]{IEEEtran}

\usepackage{url}
\usepackage{latexsym}
\usepackage{amscd,amsmath}
\usepackage{amsfonts}
\usepackage{color}
\usepackage{float}
\usepackage{longtable}
\usepackage{graphicx}
\usepackage{xspace}

\newtheorem{thm}{Theorem}[section]
\newtheorem{lem}[thm]{Lemma}
\newtheorem{cor}[thm]{Corollary}
\newtheorem{propo}[thm]{Proposition}
\newtheorem{clm}[thm]{Claim}
\newtheorem{defn}[thm]{Definition}
\newtheorem{assm}[thm]{Assumption}
\newtheorem{rem}[thm]{Remark}
\newtheorem{obs}[thm]{Observation}
\newtheorem{egs}[thm]{Example}

\newtheorem{fct}[thm]{Fact}
\newtheorem{cons}[thm]{Construction}
\newtheorem{nte}[thm]{Note}

\newenvironment{theorem}{\begin{thm}}{\end{thm}}
\newenvironment{lemma}{\begin{lem}}{\end{lem}}

\newcommand{\secref}[1]{Section~\ref{#1}}

\newcommand{\figref}[1]{Figure~\ref{#1}}

\renewcommand{\eqref}[1]{\mbox{Equation~(\ref{#1})}}

\def\figurewidth{0.97\columnwidth}
\newcommand{\widthfigure}[3]{\begin{figure}\begin{center}\begin{tabular}{|p{\figurewidth}|}\hline#1\hline\end{tabular}\end{center}\smallskip\caption{#2.}\label{#3}\end{figure}}

\ifnum\setmargin=1
  \ifnum\papertype=1
  \else
  \fi
\fi

\ifnum\llncs=0\ifnum\fullversion=1\ifnum\anonymous=0
  
\fi\fi\fi

\ifnum\llncs=0

\makeatletter
\def\appearsin#1{\gdef\@appearsin{#1}}
\def\maketitle{\par
 \begingroup
 \def\thefootnote{\arabic{footnote}}
 \def\@makefnmark{\hbox
 to 0pt{$^{\@thefnmark}$\hss}}
 \if@twocolumn
 \twocolumn[\@maketitle]
 \else \newpage
 \global\@topnum\z@ \@maketitle \fi\thispagestyle{plain}\@thanks
 \endgroup
 \setcounter{footnote}{0}
 \let\maketitle\relax
 \let\@maketitle\relax
 \gdef\@thanks{}\gdef\@author{}\gdef\@title{}\gdef\@appearsin{}
          \let\thanks\relax}
\def\@maketitle{\newpage
 \noindent \@appearsin
 \vskip 0.5in \begin{center}
 {\LARGE \@title \par} \vskip 1.5em {\large \lineskip .5em
\begin{tabular}[t]{c}\@author
 \end{tabular}\par}
 \vskip 1em {\normalsize \@date} \end{center}
 \par
 \vskip 1.5em}
\def\abstract{\if@twocolumn
\section*{Abstract}
\else \small
\begin{center}
{\bf Abstract\vspace{-.5em}\vspace{0pt}}
\end{center}
\quotation
\fi}
\def\endabstract{\if@twocolumn\else\endquotation\fi}
\mark{{}{}}
\fi

\DeclareMathAlphabet{\mathsl}{OT1}{cmr}{m}{sl}
\DeclareMathAlphabet{\mathsc}{OT1}{cmr}{m}{sc}
\DeclareMathAlphabet{\mathslbf}{OT1}{cmr}{bx}{sl}
\DeclareFontFamily{OT1}{pzc}{}
\DeclareFontShape{OT1}{pzc}{m}{it}%
             {<-> s * [1.150] pzcmi7t}{}
\DeclareMathAlphabet{\mathscript}{OT1}{pzc}{m}{it}

\newlength{\saveparindent}
\newlength{\saveparskip}
\ifnum\llncs=0

\newcount\proofqeded
\newcount\proofended
\def\qed{{\hspace{1pt}\rule[-1pt]{3pt}{9pt}}
\end{rmfamily}\addtolength{\parskip}{-0pt}
\setlength{\parindent}{\saveparindent}
\global\advance\proofqeded by 1 }
\def\qedenv{
\end{rmfamily}\addtolength{\parskip}{-0pt}
\setlength{\parindent}{\saveparindent}
\global\advance\proofqeded by 1 }
\newenvironment{proof}%
 {\proofstart}%
 {\ifnum\proofqeded=\proofended~\qed\fi \global\advance\proofended by 1
  \medskip}
 {\proofenvstart}%
 {\ifnum\proofqeded=\proofended\qedenv\fi \global\advance\proofended by 1
  \medskip}
\makeatletter
\def\proofstart{\@ifnextchar[{\@oprf}{\@nprf}}
\def\proofenvstart{\@ifnextchar[{\@osprf}{\@nsprf}}
\def\@oprf[#1]{\begin{rmfamily}\protect\vspace{6pt}\noindent{\bfseries Proof of #1:\ }%
\addtolength{\parskip}{5pt}\setlength{\parindent}{0pt}}
\def\@osprf[#1]{\begin{rmfamily}\protect\vspace{6pt}\noindent
\addtolength{\parskip}{5pt}\setlength{\parindent}{0pt}}
\def\@nprf{\begin{rmfamily}\protect\vspace{6pt}\noindent{\bfseries Proof:\ }%
\addtolength{\parskip}{5pt}\setlength{\parindent}{0pt}}
\def\@nsprf{\begin{rmfamily}\protect\vspace{6pt}\noindent%
\addtolength{\parskip}{5pt}\setlength{\parindent}{0pt}}
\fi

\newcommand{\calE}{{\cal E}}

\newcommand{\calM}{{\cal M}}

\newcommand{\calS}{{\cal S}}

\newcommand{\calX}{{\cal X}}
\newcommand{\calY}{{\cal Y}}

\newcommand{\alice}{\ensuremath{\mathsf{Alice}}\xspace}
\newcommand{\bob}{\ensuremath{\mathsf{Bob}}\xspace}

\newcommand{\f}{\ensuremath{\mathcal{F}}\xspace}
\newcommand{\s}{\ensuremath{\mathcal{S}}\xspace}
\newcommand{\env}{\ensuremath{\mathcal{Z}}\xspace}
\newcommand{\adv}{\ensuremath{\mathcal{A}}\xspace}
\newcommand{\party}{\ensuremath{\mathcal{P}}\xspace}

\newcommand{\oalice}{\ensuremath{\mathsf{output}_{\mathsf{Alice}}}\xspace}
\newcommand{\obob}{\ensuremath{\mathsf{output}_{\mathsf{Bob}}}\xspace}

\newcommand{\trans}{\ensuremath{\mathsf{trans}}\xspace}
\newcommand{\viewbob}{\ensuremath{\mathsf{view}_{\bob}}\xspace}
\newcommand{\viewalice}{\ensuremath{\mathsf{view}_{\alice}}\xspace}
\newcommand{\sid}{\ensuremath{\mathsf{sid}}\xspace}

\newcommand{\sinf}{\ensuremath{\mathrm{I_S}}\xspace}
\newcommand{\fptm}{\ensuremath{\f_{P_{V,W|X,Y}}}\xspace}
\newcommand{\frot}{\ensuremath{\f_{\mathsf{ROT}}}\xspace}
\newcommand{\fauth}{\ensuremath{\f_{\mathsf{AUTH}}}\xspace}
\DeclareMathOperator{\coinsa}{\mathsf{coins_{\alice}}}
\DeclareMathOperator{\coinsb}{\mathsf{coins_{\bob}}}

\begin{document}

\title{On the Composability of Statistically Secure Random Oblivious Transfer}

\author{Rafael~Dowsley, J\"{o}rn M\"{u}ller-Quade, Anderson~C.~A.~Nascimento
\thanks{Rafael~Dowsley is with the Department of Computer Science, Aarhus University. Email: rafael@cs.au.dk. Rafael Dowsley has received funding from the European Research Council (ERC) under the European  Union's  Horizon  2020  research  and  innovation programme  under  grant  agreement  No  669255  (MPCPRO).}
\thanks{Anderson~C.~A.~Nascimento is with the Institute of Technology, University of Washington Tacoma. E-mail: andclay@uw.edu.}
\thanks{J\"{o}rn M\"{u}ller-Quade is with the Institute of Theoretical Informatics, Karlsruhe Institute of Technology.  E-mail: mueller-quade@kit.edu.}
}

\maketitle

\begin{abstract}
We show that stand-alone statistically secure random oblivious transfer protocols based on two-party stateless primitives are statistically universally composable. I.e. they are simulatable secure with an unlimited adversary, an unlimited simulator and an unlimited environment machine. Our result implies that several previous oblivious transfer protocols in the literature which were proven secure under weaker, non-composable definitions of security can actually be used in arbitrary statistically secure applications without lowering the security.
\end{abstract}


\IEEEpeerreviewmaketitle

\section{Introduction}

Oblivious transfer (OT)~\cite{TR:Rabin81} is a primitive of central importance in modern cryptography and implies secure computation~\cite{STOC:GolMicWig87,STOC:Kilian88}. Several flavors of OT were proposed, but they are all equivalent~\cite{C:Crepeau87}. In this work we focus on the so-called one-out-of-two random oblivious transfer. This is a two-party primitive in which a sender (Alice) gets two uniformly random bits $b_0$, $b_1$ 
and a receiver (Bob) gets a uniformly random choice bit $c$ and $b_c$. Bob remains ignorant about $b_{\overline c}$. On the other hand, Alice cannot learn the choice bit $c$. 

A very large number of OT protocols are known in the stand-alone setting, based on various assumptions (both computational and physical), but this notion does not guarantee security when multiple copies of the protocol are executed, or when the OT protocols are used as building blocks within other protocols. This is an unsatisfactory state of affairs, as the major utility of OT is in the modular designing of larger protocols. Following the simulation paradigm used in~\cite{GolMicRac89} to define the seminal notion of zero-knowledge proofs of knowledge, many simulation-based definitions of security for multi-party protocols were proposed (e.g.~\cite{STOC:GolMicWig87,C:Beaver91a}) and they guarantee that the protocols are sequentially composable~\cite{JC:Canetti00}, however this paradigm of security does not guarantee general composability of the protocols. UC-security~\cite{FOCS:Canetti01} emerges as a very desirable notion of security for OT since it guarantees that the security of the protocol holds even when the OT scheme is concurrently composed with an arbitrary set of protocols.
UC-security is a very powerful notion of security that allows to fully enjoy the nice properties of OT within other protocols.

Some questions about the equivalence of stand-alone and composable security notions in the case of statistically secure protocols were studied \cite{STOC:KusLinRab06,TCC:BackMulUnr07}. In general, these security notions are not equivalent \cite{TCC:BackMulUnr07}. Therefore, it is an interesting question to study if there are restricted scenarios where this equivalence holds. 

\emph{Our Results:} In this paper we show that random OT protocols that are based on certain stateless two-party functionalities and that match a certain list of information-theoretical security properties are not only secure in a simulation-based way, but are actually UC-secure. Note that Random OT can be straightforwardly used to obtain OT for arbitrary inputs in a composable way \cite{STOC:Beaver97}. Note also that most OT protocols based on two-party stateless functionalities already internally run a random OT protocol and then use derandomization techniques to obtain OT for arbitrary inputs.
We think that this approach is interesting because, in this scenario, a protocol designer can worry only about meeting the list-based security notion and the protocol inherits the UC-security. The setting studied in this paper covers the case of statistically secure protocols based on noisy channels, cryptogates and pre-distributed correlated data. As a consequence of our result, several previously proposed protocol implementing oblivious transfer that were proven secure in weaker models automatically have their security upgraded to a simulation-based, composable one for free \cite{STOC:Beaver97,FOCS:CreKil88,EC:Crepeau97,EC:DamKilSal99,ISIT:SteWol02,SCN:CreMorWol04,IEEEIT:NasWin08,IEEEIT:PDMN11,AhlCsi13,IEEEIT:DowNas17,STOC:Kilian00,C:BeiMalMic99,Rivest99,TCC:DotKraMul11,ICITS:DowMulNil15}.

\subsection{Related Work}

OT can be constructed based both on generic computational assumptions such as the existence of enhanced trapdoor permutations~\cite{EveGolLem85,Goldreich04} and on the computational hardness of many specific problems such  as factoring~\cite{TR:Rabin81}, Diffie-Hellman~\cite{C:BelMic89,SODA:NaoPin01}, LWE \cite{C:PeiVaiWat08}, variants of LPN \cite{CANS:DavDowNas14} and McEliece assumptions~\cite{ICITS:DGMN08,IEICE:DGMN12}. However, the focus of this work is on statistically secure OT. When aiming for statistical security, OT can be based on noisy channels \cite{FOCS:CreKil88,EC:Crepeau97,EC:DamKilSal99,ISIT:SteWol02,SCN:CreMorWol04,IEEEIT:NasWin08,IEEEIT:PDMN11,AhlCsi13,IEEEIT:DowNas17}, cryptogates \cite{STOC:Kilian00,C:BeiMalMic99}, pre-distributed correlated data \cite{STOC:Beaver97,Rivest99,IEEEIT:NasWin08}, the bounded storage model~\cite{FOCS:CacCreMar98,ISIT:DowLacNas14,TCC:DHRS04,IEEEIT:DowLacNas18} and on hardware tokens~\cite{TCC:DotKraMul11,ICITS:DowMulNil15}.

Canetti and Fischlin~\cite{C:CanFis01} showed that OT cannot be UC-realized in the plain model, so additional setup assumptions are required. UC-secure OT protocols were initially constructed in the common reference string (CRS) model~\cite{STOC:CLOS02,TCC:Garay04,C:PeiVaiWat08}. In the CRS model there exists an honestly generated random string that is available to the parties (the simulator can generate its own string as long as it looks indistinguishable from the honestly generated one). In the public key infrastructure model, Damg\aa{}rd and Nielsen~\cite{C:DamNie03} proposed an OT protocol that is UC-secure against adaptive adversaries under the assumption that threshold homomorphic encryption exists. 
Katz~\cite{EC:Katz07} proved that two-party and multi-party computation are possible assuming a tamper-proof hardware. 

The question about the equivalence of stand-alone and composable security definitions for statistically secure protocols has been previously addressed in \cite{STOC:KusLinRab06,TCC:BackMulUnr07}, where it was proven that the equivalence does not hold in general.  In \cite{EC:CSSW06} it was proven that perfectly secure OT protocols according to a list of properties are \emph{sequentially} composable, this result being extended to statistical security in \cite{ICITS:CreWul08}. 

It was shown that for statistically secure commitment schemes based on two-party stateless primitives stand-alone security implies UC-security \cite{JIT:DGMN13}. While this result implies the possibility of building UC-secure OT protocols based on these commitment protocols, this is not the most efficient way of obtaining OT and it does not prove any additional security property about the existing OT protocols.

Even if the resources available to the parties to implement OT are asymmetric, Wolf and Wullschleger \cite{EC:WolWul06} showed a very simple way to reverse the OT's direction (indeed all complete two-party functionalities are reversible as proved recently by Khurana et al. \cite{EC:KKMPS16}).

\section{Preliminaries}

\subsection{Notation}

Domains of random variables will be denoted by calligraphic letters, the random variables by upper case letters and 
the realizations by lower case letters. For a random variables $X$ over $\mathcal{X}$ and $Y$ over $\mathcal{Y}$, 
$P_X: \mathcal{X} \rightarrow [0,1]$ with $\sum_{x \in \mathcal{X}} P_X(x) =1$ denotes the probability distribution of $X$, $P_X(x) := \sum_{y \in \mathcal{Y}}P_{XY}(x,y)$ the marginal probability distribution and $P_{X|Y}(x|y):=P_{XY}(x,y)/P_Y(y)$ the conditional probability distribution if $P_Y(y) \neq 0$. The statistical distance $\delta (P_X,P_Y)$ between $P_X$ and $P_Y$ with alphabet $\mathcal{X}$ is given by 
$$\delta (P_X,P_Y) = \max_{\calS \subseteq \mathcal{X}} \left|\sum_{x \in \calS} P_X(x) - P_Y(x)\right|.$$
We say $P_X$ and $P_Y$ are $\varepsilon$-close if $\delta (P_X,P_Y) \leq \varepsilon$. Following Cr\'{e}peau and Wullschleger \cite{ICITS:CreWul08}, let the statistical information of $X$ and $Y$ given $Z$ be defined as $$\sinf(X;Y|Z)= \delta (P_{XYZ}, P_Z P_{X|Z} P_{Y|Z}).$$

\subsection{The UC Framework}

Here we briefly review the main concepts of the UC framework, for more details please refer to the original work of Canetti~\cite{FOCS:Canetti01}. 
In the UC framework, the security of a protocol to carry out a certain task is ensured in three phases:

\begin{enumerate}
\item One formalizes the framework, i.e., the process of executing a protocol in the presence of an adversary and an environment machine.
\item One formalizes an ideal protocol for carrying out the task in an ideal protocol using a ``trusted party''. In the ideal protocol the trusted party captures the requirements of the desired task and the parties do not communicate among themselves.
\item One proves that the real protocol emulates the ideal protocol, i.e., for every adversary in the real model there exists an ideal adversary (also known as the simulator) in the ideal model such that no environment machine can distinguish if it is interacting with the real or the ideal world.
\end{enumerate}

The environment in the UC framework represents all activity external to the running protocol, so it provides inputs to the parties running the protocol and receives the outputs that the parties generate during the execution of the protocol. As stated above the environment also tries to distinguish between attacks on real executions of the protocol and simulated attacks against the ideal functionality. If no environment can distinguish the two situations, the real protocol emulates the ideal functionality. Proving that a protocol is secure in the UC framework provides the following benefits:

\begin{enumerate}
\item The ideal functionality describes intuitively the desired properties of the protocol.
\item The protocols are secure under composition.
\item The security is retained when the protocol is used as a sub-protocol to replace an ideal functionality that it emulates.
\end{enumerate}

\paragraph{The ideal world}
An ideal functionality $\f$ represents the desired properties of a given task. Conceptually, $\f$ is treated as a local subroutine by the several parties that use it, and so the communication between the parties and $\f$ is supposedly secure (i.e., messages are sent by input and output tapes). 
The ideal protocol also involves a simulator $\s$, an environment $\env$ on input $z$ and a set of dummy parties that interacts as defined below. Whenever a dummy party is activated with input $x$, it writes $x$ onto the input tape of $\f$. Whenever the dummy party is activated with value $x$ on its subroutine output tape, it writes $x$ on subroutine output tape of $\env$. The simulator $\s$ has no access to the contents of messages sent between dummy parties and $\f$, and it should send corruption messages directly to $\f$, who is responsible for determining the effects of corrupting any dummy party. The ideal functionality receives messages from the dummy parties by reading its input tape and sends messages to them by writing to their subroutine output tape. In the ideal protocol there is no communication among the parties. The environment $\env$ can set the inputs to the parties and read their outputs, but cannot see the communication with the ideal functionality.

\paragraph{The real world}
In the real world, the protocol $\pi$ is executed by parties $\party_1,\ldots, \party_n$ with some adversary $\adv$ and an environment machine $\env$ with input $z$. $\env$ can set the inputs for the parties and see their outputs, but not the communication among the parties. The parties can invoke subroutines, pass inputs to them and receive outputs from them. They can also write messages on the incoming communication tape of the adversary. These messages may specify the identity of the final destination of the message. $\adv$ can send messages to any party ($\adv$ delivers the message). In addition, they may use the ideal functionalities that are provided to the real protocol.
$\adv$ can communicate with $\env$ and the ideal functionalities that are provided to the real protocol. $\adv$ also controls the corrupt parties (the environment always knows which parties are corrupted).

\widthfigure{
\begin{center}
\textbf{Functionality} $\frot{}$
\smallskip
\end{center}\\
$\frot{}$ interacts with $\alice$ and $\bob$.\\
\\
\textbf{Alice's Check-in:} Upon receiving (\textsc{Distribute}, $\sid$, \ldots) from $\alice$, if $\alice$ is honest sample uniformly random $b_0,b_1 \in \{0,1\}$; otherwise set the bits $b_0,b_1$ as specified in $\alice$'s message. 
Record $(\sid, b_0,b_1)$ and ignore future (\textsc{Distribute}, $\sid$, \ldots) from $\alice$.\\
\\

\textbf{Bob's Check-in:} Upon receiving (\textsc{Distribute}, $\sid$, \ldots) from $\bob$, if $\bob$ is honest sample a uniformly random $c \in \{0,1\}$; otherwise set the bit $c$ as specified in $\bob$'s message. Record $(\sid, c)$
and ignore future (\textsc{Distribute}, $\sid$, \ldots) from $\bob$. \\
\\

\textbf{Distribution:}  Upon having recorded values $b_0$, $b_1$ and $c$ for some $\sid$, send (\textsc{Output}, $\sid$) to $\s$. Upon an answer (\textsc{Output}, $\sid$) from $\s$, deliver (\textsc{Output}, $b_0$, $b_1$) to $\alice$ and (\textsc{Output}, $c$, $b_c$) to $\bob$.\\
\\
}{The one-out-of-two bit random oblivious transfer functionality}{fig-fot}

\paragraph{The adversarial model}
The network is asynchronous without guaranteed delivery of messages. The communication is public, but authenticated (i.e., the adversary cannot modify the messages). The adversary is active in its control over corrupted parties. Any number of parties can be corrupted. Finally, the adversary, the environment and the simulator are allowed unbounded complexity. This assumption on the computational power of the simulator somehow weakens our result as the composition theorem cannot be applied several times if the real adversary were restricted to polynomial time, because the ``is at least as secure as'' relation cannot be proven to be transitive anymore. However, arbitrary composition is allowed when considering statistically secure protocols and this situation is  common in the literature when proving general results on the composability of statistically secure protocols \cite{TCC:BackMulUnr07,STOC:KusLinRab06,ICITS:CreWul08,EC:CSSW06}.

\paragraph{Realizing an ideal functionality}

A protocol $\pi$ statistically UC-realizes an ideal functionality $\f$ if for any real-life adversary $\adv$ there exists a simulator $\s$ such that no environment $\env$, on any input $z$, can tell with non-negligible probability whether it is interacting with $\adv$ and parties running $\pi$ in the real-life process, or it is interacting with $\s$ and $\f$ in the ideal protocol. This means that, from the point of view of the environment, running protocol $\pi$ is statistically indistinguishable from the ideal world with $\f$.

\paragraph{The Oblivious Transfer Functionality}

We present in \figref{fig-fot} the one-out-of-two bit random oblivious transfer functionality $\frot$. The sender will be denote by $\alice$ and the receiver by $\bob$. 

\subsection{Setup Assumption}\label{sec:fbsc}

In this work we consider the scenario in which $\alice$ and $\bob$ have access to the functionality $\fptm$ that given inputs $x \in \mathcal{X}$ from Alice and $y \in \mathcal{Y}$ from Bob samples the outputs $v \in \mathcal{V}$ and $w \in \mathcal{W}$ according to the conditional probability distribution $P_{V,W|X,Y}$, and gives the outputs $v$ and $w$ to $\alice$ and $\bob$, respectively. The functionality $\fptm$ is described in \figref{fig-fptm}. Note that this functionality captures setup assumptions that are commonly used for obtaining statistically secure OT protocols, such as the existence of a stateless noisy channel between the parties, cryptogates and pre-distributed correlated data.

\widthfigure{
\begin{center}
\textbf{Functionality} $\fptm{}$
\smallskip
\end{center}\\
$\fptm{}$ interacts with $\alice$ and $\bob$ and is parametrized by the conditional probability distribution $P_{V,W|X,Y}$.\\
\\
\textbf{Alice's Input:} Upon receiving (\textsc{Input}, $\sid$, $x$) from $\alice$, if $x \in \mathcal{X}$ then record ($\sid$, $x$). Ignore future messages (\textsc{Input}, $\sid$, \ldots) from $\alice$.\\
\\
\textbf{Bob's Input:} Upon receiving (\textsc{Input}, $\sid$, $y$) from $\bob$, if $y \in \mathcal{Y}$ then record ($\sid$, $y$).
Ignore future messages (\textsc{Input}, $\sid$, \ldots) from $\bob$.\\
\\
\textbf{Output:} Upon obtained valid inputs from $\alice$ and $\bob$ for some $\sid$, pick $v, w$ according to $P_{V,W|X,Y}$ and output (\textsc{Output}, $\sid$, $v$) to $\alice$ and (\textsc{Output}, $\sid$, $w$) to $\bob$.\\
\\
}{The functionality that given valid inputs, samples outputs according to the conditional probability distribution and delivers the outputs to $\alice$ and $\bob$}{fig-fptm}

\section{Random Oblivious Transfer Based on Statistically Secure Two Party Stateless Functionalities}\label{sec:model}

In this section we define a stand-alone security model for random OT protocols that achieve statistical security by using $\fptm{}$ as a setup assumption. $\alice$ and $\bob$  have two resources available between them:

\begin{itemize}
\item a bidirectional authenticated noiseless channel denoted as $\fauth$ and
\item the functionality $\fptm{}$.
\end{itemize}
We model the probabilistic choices of $\alice$ by a random variable $\coinsa$ and those of $\bob$ by a random variable $\coinsb$, so that we can use deterministic functions in the protocol. 
As usual, we assume that the noiseless messages exchanged by the players and their personal randomness are taken from $\{0,1\}^*$.

\paragraph{Protocol $\pi$} $\alice$ and $\bob$ interact and in the end of the execution $\alice$ gets $(b_0, b_1)$ and $\bob$ gets $(c, b_c)$, for $b_0, b_1,c \in \{0, 1\}$ picked uniformly at random.
The security parameter is $n$, and determines how many times the parties can use the functionality $\fptm{}$: in the $i$-th round $\alice$ and $\bob$ input symbols $x_i$ and $y_i$ to the functionality $\fptm{}$, which generates the outputs $v_i$ and $w_i$ according to $P_{V,W|X,Y}$ and delivers them to $\alice$ and $\bob$, respectively. Let $x^i$, $y^i$, $v^i$ and $w^i$ denote the vectors of these variables until $i$-th round. The parties can use $\fauth$ at any moment. Let $\trans$ denote all the noiseless messages exchanged between the players.

We call the view of $\bob$ all the data in his possession, i.e. $y^{n}, w^n, c, \allowbreak \coinsb$ and $\trans$, and denote it by $\viewbob$. $\viewalice$ is defined similarly. We denote the output of the (possibly malicious) parties $\alice$ and $\bob$ by $\oalice$ and $\obob$, respectively. The stand-alone definition of security that is henceforth considered in this paper follows the lines of Cr\'{e}peau and Wullschleger~\cite{ICITS:CreWul08}. The protocol is said to be secure if there exists an $\epsilon$ that is a negligible function of the security parameter $n$ and is such that the following properties are satisfied:

\paragraph{Correctness} If both parties are honest, then $\oalice= (b_0, b_1)$ and $\obob=(c, d)$ for $d \in \{0, 1\}$ and uniformly random $b_0, b_1,c \in \{0, 1\}$. Additionally,
$$\Pr[D=B_C]\geq 1 - \epsilon.$$

\paragraph{Security for $\alice$} If $\alice$ is honest, then $\oalice= (b_0, b_1)$ for uniformly random $b_0, b_1 \in \{0, 1\}$ and there exists a random variable $C$ such that $$I_S(B_0,B_1;C) \leq \epsilon$$
and  $$I_S(B_0,B_1;\obob|C,B_{C}) \leq \epsilon.$$

\paragraph{Security for $\bob$} If $\bob$ is honest, then $\obob=(c, d)$ for $d \in \{0, 1\}$ and uniformly random $c \in \{0, 1\}$; and 
$$I_S(C;\oalice) \leq \epsilon .$$

\section{UC-Security Implication}

In this section we address the question of whether random OT protocols that are secure according to the definitions of \secref{sec:model} also enjoy statistical UC-security. We will show that this is indeed the case.
Intuitively this follows from the fact that the security in those protocols is based on the correlated randomness that is provided by the functionality $\fptm{}$ to $\alice$ and $\bob$. Since in the ideal world the simulator controls 
 $\fptm{}$, it can leverage this knowledge in order to extract the outputs of the corrupted parties and forward them to the random oblivious transfer functionality $\frot{}$, thus allowing the ideal execution to be indistinguishable from the real execution from the environment's point of view. First we prove some lemmas that will be used later on to prove the main result of this work.

We first show that in any random OT protocol that is stand-alone secure, if $\bob$ is honest, then given $\alice$'s input to and output from the functionality $\fptm{}$
and all the noiseless communication exchanged by $\alice$ and $\bob$ through $\fauth$ it is possible to extract both outputs that $\bob$ would get with $c=0$ and $c=1$ in the random OT protocol.

\begin{lemma}
\label{lem:exs}
Let $\pi$ be a stand-alone secure random OT protocol and let $\bob$ be honest. Given $\alice$'s input to and output from $\fptm{}$ and all the noiseless communication exchanged by $\alice$ and $\bob$ through $\fauth$ during the execution of $\pi$, with overwhelming probability it is possible to extract the output that $\bob$ would get both in the case that $c=0$ and $c=1$.
\end{lemma}

\begin{proof}
Lets consider an execution of the protocol $\pi$ in which $\bob$ has random coins $\coinsb$ and gets $\obob=(c, d)$ for $d \in \{0, 1\}$ and uniformly random $c \in \{0, 1\}$ (as the protocol is stand-alone secure).
Denote by $m$ the set of messages exchanged between $\alice$ and $\fptm{}$ concatenated with the noiseless messages  between $\bob$ and $\alice$.  We claim that there should exist $\overline{\coinsb} \neq \coinsb$ so that for the same $m$,  if $\bob$ executed the protocol with $\overline{\coinsb}$ he should have been able with overwhelming probability to get $\overline{\obob}=(\overline{c},\overline{d})$ with $\overline{c} \neq c$ and
$\overline{d} \in \{0, 1\}$. If that were not the case,  $\alice$ would know that $\bob$ is unable to obtain a valid output $\overline{d}$ when the choice bit is $\overline{c}$, thus gaining knowledge on the choice bit and breaking the protocol security. Given that
$$I_S(\oalice;C) \leq \epsilon ,$$
we get
$$\delta (P_{\oalice C},  P_{\oalice} P_{C})\leq \epsilon,$$
and so there are events $\calE_1$ and $\calE_2$ such that 
$$\Pr[\calE_1]=\Pr[\calE_2]= 1 - \epsilon \text{ and}$$
$$P_{\oalice C|\calE_1}=  P_{\oalice|\calE_2} P_{C|\calE_2}.$$
Therefore if $\calE_1$ and $\calE_2$ happen, then $\oalice$ does not provide information about $C$ and $\overline{\coinsb}$ should exist. Thus, given $m$ we are left with an extraction procedure. One just computes $\coinsb$ and $\overline{\coinsb}$ that for this $m$ produce outputs $\obob$ and $\overline{\obob}$, respectively, and simulates the protocol execution for each specified $\bob$'s randomness.\end{proof}

We now prove that given access to the messages that $\bob$ exchanges with $\alice$ and $\fptm{}$, there is a point in the protocol execution in which it is possible to extract the choice bit $c$ and still equivocate  
$b_0,b_1$ to any value, i.e., it is possible to find an $\alice$' view that is compatible with the current view of $\bob$ and the new values of $b_0$ and $b_1$. 

\begin{lemma}
\label{lem:exr}
Let $\pi$ be a stand-alone secure random OT protocol and let $\alice$ be honest. Given access to all messages that $\bob$'s exchanges with $\fptm{}$ and all the noiseless communication exchanged by $\alice$ and $\bob$ through $\fauth$ during the execution of $\pi$, with overwhelming probability it is possible to extract the choice bit $c$ at some point of the execution of the protocol $\pi$. Additionally at this point it is still possible to change $b_0$ and $b_1$ to any desired values.
\end{lemma}

\begin{proof}
We first prove that there is a point in the protocol execution where we can extract the choice bit given the messages that $\bob$ exchanged with $\alice$ and the functionality $\fptm{}$. Let $m$ denote these messages in a given protocol execution. Let $\calM(0)$ denote the set of messages that allow $\bob$ to obtain the bit $b_0$ with overwhelming probability (the probability taken over $\coinsa$, $\coinsb$ and the randomness of $\fptm{}$). And let $\calM(1)$ be defined similarly for $b_1$. From the stand-alone security for Alice we have that

$$I_S(B_0,B_1;C) \leq \epsilon$$ and  $$I_S(B_0,B_1;\obob|C,B_{C}) \leq \epsilon,$$ and so we get that $m$ with overwhelming probability (over $\coinsa$ and the randomness of $\fptm{}$) cannot be in both $\calM(0)$ and $\calM(1)$, since this fact would imply that the resulting protocol would be insecure for $\alice$. This fact gives us a procedure for obtaining the choice bit $c$ given $m$. We just check if $m$ is in $\calM(0)$ or $\calM(1)$. 

We now turn to the equivocation property. From the previous reasoning, we know that there should exist a point in the protocol where $\bob$ sends a message to $\alice$ that fixes the choice bit (i.e. the choice bit can be extracted from his messages from/to $\alice$ and $\fptm{}$). Let $i$ be the index of such message. Suppose the $i$-th message is the very last one in the protocol. Then $\bob$ has all the information necessary to compute his output even before sending the $i$-th message. As the choice bit is only fixed in the next message, $\bob$ should be able to compute both $b_0$ and $b_1$, breaking $\alice$'s security.
Thus, the $i$-th message should not be the last one. The same reasoning implies that from $\bob$'s point of view, none of $\alice$'s outputs $b_0$ and $b_1$ can be fixed before the $i$-th message: (1) if both $b_0$ and $b_1$ are fixed from $\bob$'s point of view before the $i$-th message, then he could obtain both $b_0$ and $b_1$ and break the stand-alone security; (2) if only $b_i$ is fixed, then $\bob$ can still change his choice to $c=1-i$ and obtain both $b_0$ and $b_1$, thus breaking the stand-alone security. Therefore, we should have that when the $i$-th message is sent by $\bob$, $\alice$'s outputs $b_0$ and $b_1$ are still equivocable.\end{proof}

We now use two lemmas to prove our main result:
\begin{theorem}
\label{the:uc}
Any stand-alone statistically secure protocol $\pi$ of random oblivious transfer based on $\fauth$ and $\fptm{}$ UC-realizes $\frot{}$.
\end{theorem}

\begin{proof}
We construct the simulator $\s$ as follows. $\s$ runs a simulated copy of $\adv$ in a black-box way, plays the role of the ideal functionality $\fptm{}$ and simulates a copy of the hybrid interaction of $\pi$ for the simulated adversary $\adv$. In addition, $\s$ forwards the messages between $\env$ and $\adv$. Below we describe the procedures of the simulator in each occasion:

\textbf{Only $\alice$ is corrupted:} $\s$ samples the randomness $\coinsb$ of the simulated $\bob$ and proceeds with the simulated execution of the protocol $\pi$ by producing his noiseless messages as well as his inputs $y_i \in \calY$ to $\fptm{}$. Additionally, once the inputs $x_i \in \calX$ and $y_i \in \calY$ to $\fptm{}$ are fixed, $\s$ simulates the outputs of the functionality $\fptm{}$ and sends $v_i$ to $\adv$. As $\s$ plays the role of $\fptm{}$, when the execution is done, $\s$ extracts the output bits $b_0,b_1$ of the corrupted $\alice$ using the result of lemma~\ref{lem:exs} and forwards $b_0,b_1$ to $\frot{}$. $\s$ then allows $\frot{}$ to deliver the  output.
\\

\textbf{Only $\bob$ is corrupted:} $\s$ samples the randomness $\coinsa$ of the simulated $\alice$ and proceeds with the simulated execution of the protocol $\pi$ by producing her noiseless messages as well as her inputs $x_i \in \calX$ to $\fptm{}$. Additionally, once the inputs $x_i \in \calX$ and $y_i \in \calY$ to $\fptm{}$ are fixed, $\s$ simulates the outputs of the functionality $\fptm{}$ and sends $w_i$ to $\adv$. 
Then using the result of lemma~\ref{lem:exr}, $\s$ extracts the choice bit $c$ of the corrupted $\bob$, inputs $c$ to $\frot{}$, receives $b_c$ and finishes the simulated protocol execution in such way that the received bit in the hybrid interaction $b'_c$ is equal to the received bit in the ideal protocol $b_c$ with overwhelming probability.
\\

\textbf{Neither party is corrupted:} $\s$ samples the randomness $\coinsa$ and $\coinsb$ and proceeds with the simulated execution of the protocol $\pi$ by
simulating the noiseless messages as well as the inputs/outputs of $\fptm{}$, and reveals the noiseless messages to $\adv$. If the simulated $\bob$ would output $b'_{c'}$ in the hybrid interaction, then $\s$ allows $\frot{}$ to output the bit $b_c$.
\\

\textbf{Both parties are corrupted:} $\s$ just simulates $\fptm{}$.
\\

We analyze below the probabilities of the events that can result in different views for the environment $\env$ between the real world execution with the protocol $\pi$ and the adversary $\adv$, and the ideal world execution with functionality $\frot{}$ and 
the simulator $\s$:

\begin{itemize}
	\item When only $\alice$ is corrupted, $\env$'s view in the real and ideal worlds are equal if: (1) $\s$ succeeds to extract both of $\alice$'s outputs bits $b_0,b_1$ to forward to $\frot{}$; (2) $\adv$ does not learn the choice bit $c'$ in the simulated protocol execution. By lemma~\ref{lem:exs}, the extraction works with overwhelming probability. By the stand-alone security, with overwhelming probability $\adv$ does not learn $c'$.
		
	\item When only $\bob$ is corrupted, $\env$'s view in the real and ideal worlds are equal if: (1) $\s$ succeeds to extract the bit $c$ and finish the protocol in such way that the received bit $b'_c$ in the simulated protocol execution is equal to $b_c$; (2) $\adv$ cannot learn $b'_{\overline c}$ in the simulated protocol execution. By lemma~\ref{lem:exr}, the first condition is satisfied with overwhelming probability. By the stand-alone security, with overwhelming probability $\adv$ cannot learn $b'_{\overline c}$
	
	\item When neither party is corrupted, $\s$'s procedures statistically emulate the hybrid execution for the adversary $\adv$, as $\adv$ cannot learn $b'_0,b'_1,c'$ from the noiseless messages alone.
	\item When both parties are corrupted, $\s$'s procedures perfectly emulate the hybrid execution for the adversary $\adv$.	
\end{itemize}

We conclude that since all events that can result in different views have negligible probabilities, the protocol $\pi$ UC-realizes $\frot{}$.
\end{proof}

\section{Conclusion}

In this paper, we prove that random oblivious transfer protocols based on two-party stateless functionalities matching a list of security properties are universally composable when unbounded simulators are allowed. As previously commented, this assumption on the simulator gives us secure universal composability with other statistically secure protocols. The restriction to random oblivious transfer protocols is not restrictive (since random OT can be used to obtain OT for arbitrary inputs \cite{STOC:Beaver97}, proving the composability of such reduction is straightforward). And most of the OT protocols based on two-party stateless functionalities are in fact designed to initially run an internal random OT protocol and then derandomize the values. In this case the universally composability implication can be applied directly to the inner random OT protocol. However, it is an interesting problem to generalize the results presented here to arbitrary OT. Our result immediately imply that several previously proposed OT protocols can have their security upgraded for free \cite{FOCS:CreKil88,EC:Crepeau97,EC:DamKilSal99,ISIT:SteWol02,SCN:CreMorWol04,IEEEIT:NasWin08,IEEEIT:PDMN11,AhlCsi13,IEEEIT:DowNas17,STOC:Kilian00,C:BeiMalMic99,STOC:Beaver97,Rivest99,TCC:DotKraMul11,ICITS:DowMulNil15}.





\begin{thebibliography}{10}
\providecommand{\url}[1]{#1}
\csname url@samestyle\endcsname
\providecommand{\newblock}{\relax}
\providecommand{\bibinfo}[2]{#2}
\providecommand{\BIBentrySTDinterwordspacing}{\spaceskip=0pt\relax}
\providecommand{\BIBentryALTinterwordstretchfactor}{4}
\providecommand{\BIBentryALTinterwordspacing}{\spaceskip=\fontdimen2\font plus
\BIBentryALTinterwordstretchfactor\fontdimen3\font minus
  \fontdimen4\font\relax}
\providecommand{\BIBforeignlanguage}[2]{{%
\expandafter\ifx\csname l@#1\endcsname\relax
\typeout{** WARNING: IEEEtran.bst: No hyphenation pattern has been}%
\typeout{** loaded for the language `#1'. Using the pattern for}%
\typeout{** the default language instead.}%
\else
\language=\csname l@#1\endcsname
\fi
#2}}
\providecommand{\BIBdecl}{\relax}
\BIBdecl

\bibitem{TR:Rabin81}
M.~O. Rabin, ``How to exchange secrets by oblivious transfer,'' Aiken
  Computation Laboratory, Harvard University, Tech. Rep. Technical Memo TR-81,
  1981.

\bibitem{STOC:GolMicWig87}
O.~Goldreich, S.~Micali, and A.~Wigderson, ``How to play any mental game or {A}
  completeness theorem for protocols with honest majority,'' in \emph{19th
  Annual {ACM} Symposium on Theory of Computing}, A.~Aho, Ed.\hskip 1em plus
  0.5em minus 0.4em\relax {ACM} Press, May 1987, pp. 218--229.

\bibitem{STOC:Kilian88}
J.~Kilian, ``Founding cryptography on oblivious transfer,'' in \emph{20th
  Annual {ACM} Symposium on Theory of Computing}.\hskip 1em plus 0.5em minus
  0.4em\relax {ACM} Press, May 1988, pp. 20--31.

\bibitem{C:Crepeau87}
C.~Cr{\'e}peau, ``Equivalence between two flavours of oblivious transfers,'' in
  \emph{Advances in Cryptology -- {CRYPTO}'87}, ser. Lecture Notes in Computer
  Science, C.~Pomerance, Ed., vol. 293.\hskip 1em plus 0.5em minus 0.4em\relax
  Springer, Heidelberg, Aug. 1988, pp. 350--354.

\bibitem{GolMicRac89}
S.~Goldwasser, S.~Micali, and C.~Rackoff, ``The knowledge complexity of
  interactive proof systems,'' \emph{{SIAM} Journal on Computing}, vol.~18,
  no.~1, pp. 186--208, 1989.

\bibitem{C:Beaver91a}
D.~Beaver, ``Foundations of secure interactive computing,'' in \emph{Advances
  in Cryptology -- {CRYPTO}'91}, ser. Lecture Notes in Computer Science,
  J.~Feigenbaum, Ed., vol. 576.\hskip 1em plus 0.5em minus 0.4em\relax
  Springer, Heidelberg, Aug. 1992, pp. 377--391.

\bibitem{JC:Canetti00}
R.~Canetti, ``Security and composition of multiparty cryptographic protocols,''
  \emph{Journal of Cryptology}, vol.~13, no.~1, pp. 143--202, 2000.

\bibitem{FOCS:Canetti01}
------, ``Universally composable security: A new paradigm for cryptographic
  protocols,'' in \emph{42nd Annual Symposium on Foundations of Computer
  Science}.\hskip 1em plus 0.5em minus 0.4em\relax {IEEE} Computer Society
  Press, Oct. 2001, pp. 136--145.

\bibitem{STOC:KusLinRab06}
E.~Kushilevitz, Y.~Lindell, and T.~Rabin, ``Information-theoretically secure
  protocols and security under composition,'' in \emph{38th Annual {ACM}
  Symposium on Theory of Computing}, J.~M. Kleinberg, Ed.\hskip 1em plus 0.5em
  minus 0.4em\relax {ACM} Press, May 2006, pp. 109--118.

\bibitem{TCC:BackMulUnr07}
M.~Backes, J.~M{\"u}ller-Quade, and D.~Unruh, ``On the necessity of rewinding
  in secure multiparty computation,'' in \emph{TCC~2007: 4th Theory of
  Cryptography Conference}, ser. Lecture Notes in Computer Science, S.~P.
  Vadhan, Ed., vol. 4392.\hskip 1em plus 0.5em minus 0.4em\relax Springer,
  Heidelberg, Feb. 2007, pp. 157--173.

\bibitem{STOC:Beaver97}
D.~Beaver, ``Commodity-based cryptography (extended abstract),'' in \emph{29th
  Annual {ACM} Symposium on Theory of Computing}.\hskip 1em plus 0.5em minus
  0.4em\relax {ACM} Press, May 1997, pp. 446--455.

\bibitem{FOCS:CreKil88}
C.~Cr{\'e}peau and J.~Kilian, ``Achieving oblivious transfer using weakened
  security assumptions (extended abstract),'' in \emph{29th Annual Symposium on
  Foundations of Computer Science}.\hskip 1em plus 0.5em minus 0.4em\relax
  {IEEE} Computer Society Press, Oct. 1988, pp. 42--52.

\bibitem{EC:Crepeau97}
C.~Cr{\'e}peau, ``Efficient cryptographic protocols based on noisy channels,''
  in \emph{Advances in Cryptology -- {EUROCRYPT}'97}, ser. Lecture Notes in
  Computer Science, W.~Fumy, Ed., vol. 1233.\hskip 1em plus 0.5em minus
  0.4em\relax Springer, Heidelberg, May 1997, pp. 306--317.

\bibitem{EC:DamKilSal99}
I.~Damg{\aa}rd, J.~Kilian, and L.~Salvail, ``On the (im)possibility of basing
  oblivious transfer and bit commitment on weakened security assumptions,'' in
  \emph{Advances in Cryptology -- {EUROCRYPT}'99}, ser. Lecture Notes in
  Computer Science, J.~Stern, Ed., vol. 1592.\hskip 1em plus 0.5em minus
  0.4em\relax Springer, Heidelberg, May 1999, pp. 56--73.

\bibitem{ISIT:SteWol02}
D.~Stebila and S.~Wolf, ``Efficient oblivious transfer from any non-trivial
  binary-symmetric channel,'' in \emph{Information Theory, 2002. Proceedings.
  2002 IEEE International Symposium on}, Lausanne, Switzerland,
  Jun.~30~--~Jul.~5, 2002, p. 293.

\bibitem{SCN:CreMorWol04}
C.~Cr{\'e}peau, K.~Morozov, and S.~Wolf, ``Efficient unconditional oblivious
  transfer from almost any noisy channel,'' in \emph{SCN 04: 4th International
  Conference on Security in Communication Networks}, ser. Lecture Notes in
  Computer Science, C.~Blundo and S.~Cimato, Eds., vol. 3352.\hskip 1em plus
  0.5em minus 0.4em\relax Springer, Heidelberg, Sep. 2005, pp. 47--59.

\bibitem{IEEEIT:NasWin08}
A.~C.~A. Nascimento and A.~Winter, ``On the oblivious-transfer capacity of
  noisy resources,'' \emph{Information Theory, IEEE Transactions on}, vol.~54,
  no.~6, pp. 2572--2581, Jun. 2008.

\bibitem{IEEEIT:PDMN11}
A.~C.~B. Pinto, R.~Dowsley, K.~Morozov, and A.~C.~A. Nascimento, ``Achieving
  oblivious transfer capacity of generalized erasure channels in the malicious
  model,'' \emph{Information Theory, IEEE Transactions on}, vol.~57, no.~8, pp.
  5566--5571, Aug. 2011.

\bibitem{AhlCsi13}
R.~Ahlswede and I.~Csisz\'ar, ``On oblivious transfer capacity,'' in
  \emph{Information Theory, Combinatorics, and Search Theory}, ser. Lecture
  Notes in Computer Science, H.~Aydinian, F.~Cicalese, and C.~Deppe, Eds.\hskip
  1em plus 0.5em minus 0.4em\relax Springer Berlin Heidelberg, 2013, vol. 7777,
  pp. 145--166.

\bibitem{IEEEIT:DowNas17}
R.~Dowsley and A.~C.~A. Nascimento, ``On the oblivious transfer capacity of
  generalized erasure channels against malicious adversaries: The case of low
  erasure probability,'' \emph{IEEE Transactions on Information Theory},
  vol.~63, no.~10, pp. 6819--6826, Oct 2017.

\bibitem{STOC:Kilian00}
J.~Kilian, ``More general completeness theorems for secure two-party
  computation,'' in \emph{32nd Annual {ACM} Symposium on Theory of
  Computing}.\hskip 1em plus 0.5em minus 0.4em\relax {ACM} Press, May 2000, pp.
  316--324.

\bibitem{C:BeiMalMic99}
A.~Beimel, T.~Malkin, and S.~Micali, ``The all-or-nothing nature of two-party
  secure computation,'' in \emph{Advances in Cryptology -- {CRYPTO}'99}, ser.
  Lecture Notes in Computer Science, M.~J. Wiener, Ed., vol. 1666.\hskip 1em
  plus 0.5em minus 0.4em\relax Springer, Heidelberg, Aug. 1999, pp. 80--97.

\bibitem{Rivest99}
R.~L. Rivest, ``Unconditionally secure commitment and oblivious transfer
  schemes using private channels and a trusted initializer,'' 1999, preprint
  available at http://people.csail.mit.edu/rivest/Rivest- commitment.pdf.

\bibitem{TCC:DotKraMul11}
N.~D{\"o}ttling, D.~Kraschewski, and J.~M{\"u}ller-Quade, ``Unconditional and
  composable security using a single stateful tamper-proof hardware token,'' in
  \emph{TCC~2011: 8th Theory of Cryptography Conference}, ser. Lecture Notes in
  Computer Science, Y.~Ishai, Ed., vol. 6597.\hskip 1em plus 0.5em minus
  0.4em\relax Springer, Heidelberg, Mar. 2011, pp. 164--181.

\bibitem{ICITS:DowMulNil15}
R.~Dowsley, J.~M{\"u}ller-Quade, and T.~Nilges, ``Weakening the isolation
  assumption of tamper-proof hardware tokens,'' in \emph{ICITS 15: 8th
  International Conference on Information Theoretic Security}, ser. Lecture
  Notes in Computer Science, A.~Lehmann and S.~Wolf, Eds., vol. 9063.\hskip 1em
  plus 0.5em minus 0.4em\relax Springer, Heidelberg, May 2015, pp. 197--213.

\bibitem{EveGolLem85}
\BIBentryALTinterwordspacing
S.~Even, O.~Goldreich, and A.~Lempel, ``A randomized protocol for signing
  contracts,'' \emph{Commun. ACM}, vol.~28, no.~6, pp. 637--647, Jun. 1985.
  [Online]. Available: \url{http://doi.acm.org/10.1145/3812.3818}
\BIBentrySTDinterwordspacing

\bibitem{Goldreich04}
O.~Goldreich, \emph{Foundations of Cryptography: Basic Applications}.\hskip 1em
  plus 0.5em minus 0.4em\relax Cambridge, UK: Cambridge University Press, 2004,
  vol.~2.

\bibitem{C:BelMic89}
M.~Bellare and S.~Micali, ``Non-interactive oblivious transfer and
  spplications,'' in \emph{Advances in Cryptology -- {CRYPTO}'89}, ser. Lecture
  Notes in Computer Science, G.~Brassard, Ed., vol. 435.\hskip 1em plus 0.5em
  minus 0.4em\relax Springer, Heidelberg, Aug. 1990, pp. 547--557.

\bibitem{SODA:NaoPin01}
M.~Naor and B.~Pinkas, ``Efficient oblivious transfer protocols,'' in
  \emph{12th Annual {ACM}-{SIAM} Symposium on Discrete Algorithms}, S.~R.
  Kosaraju, Ed.\hskip 1em plus 0.5em minus 0.4em\relax {ACM-SIAM}, Jan. 2001,
  pp. 448--457.

\bibitem{C:PeiVaiWat08}
C.~Peikert, V.~Vaikuntanathan, and B.~Waters, ``A framework for efficient and
  composable oblivious transfer,'' in \emph{Advances in Cryptology --
  {CRYPTO}~2008}, ser. Lecture Notes in Computer Science, D.~Wagner, Ed., vol.
  5157.\hskip 1em plus 0.5em minus 0.4em\relax Springer, Heidelberg, Aug. 2008,
  pp. 554--571.

\bibitem{CANS:DavDowNas14}
B.~David, R.~Dowsley, and A.~C.~A. Nascimento, ``Universally composable
  oblivious transfer based on a variant of {LPN},'' in \emph{CANS 14: 13th
  International Conference on Cryptology and Network Security}, ser. Lecture
  Notes in Computer Science, D.~Gritzalis, A.~Kiayias, and I.~G. Askoxylakis,
  Eds., vol. 8813.\hskip 1em plus 0.5em minus 0.4em\relax Springer, Heidelberg,
  Oct. 2014, pp. 143--158.

\bibitem{ICITS:DGMN08}
R.~Dowsley, J.~van~de Graaf, J.~{M{\"u}ller-Quade}, and A.~C.~A. Nascimento,
  ``Oblivious transfer based on the {McEliece} assumptions,'' in \emph{ICITS
  08: 3rd International Conference on Information Theoretic Security}, ser.
  Lecture Notes in Computer Science, R.~Safavi-Naini, Ed., vol. 5155.\hskip 1em
  plus 0.5em minus 0.4em\relax Springer, Heidelberg, Aug. 2008, pp. 107--117.

\bibitem{IEICE:DGMN12}
R.~Dowsley, J.~van~de Graaf, J.~M{\"u}ller-Quade, and A.~C.~A. Nascimento,
  ``Oblivious transfer based on the {McEliece} assumptions,'' \emph{IEICE
  Transactions on Fundamentals of Electronics, Communications and Computer
  Sciences}, vol. E95-A, no.~2, pp. 567--575, 2012.

\bibitem{FOCS:CacCreMar98}
C.~Cachin, C.~Cr{\'e}peau, and J.~Marcil, ``Oblivious transfer with a
  memory-bounded receiver,'' in \emph{39th Annual Symposium on Foundations of
  Computer Science}.\hskip 1em plus 0.5em minus 0.4em\relax {IEEE} Computer
  Society Press, Nov. 1998, pp. 493--502.

\bibitem{ISIT:DowLacNas14}
R.~Dowsley, F.~Lacerda, and A.~C.~A. Nascimento, ``Oblivious transfer in the
  bounded storage model with errors,'' in \emph{Information Theory (ISIT), 2014
  IEEE International Symposium on}, Honolulu, HI, USA, Jun.~29~--~Jul.~4, 2014,
  pp. 1623--1627.

\bibitem{TCC:DHRS04}
Y.~Z. Ding, D.~Harnik, A.~Rosen, and R.~Shaltiel, ``Constant-round oblivious
  transfer in the bounded storage model,'' in \emph{TCC~2004: 1st Theory of
  Cryptography Conference}, ser. Lecture Notes in Computer Science, M.~Naor,
  Ed., vol. 2951.\hskip 1em plus 0.5em minus 0.4em\relax Springer, Heidelberg,
  Feb. 2004, pp. 446--472.

\bibitem{IEEEIT:DowLacNas18}
R.~Dowsley, F.~Lacerda, and A.~C.~A. Nascimento, ``Commitment and oblivious
  transfer in the bounded storage model with errors,'' \emph{IEEE Transactions
  on Information Theory}, vol.~64, no.~8, pp. 5970--5984, Aug 2018.

\bibitem{C:CanFis01}
R.~Canetti and M.~Fischlin, ``Universally composable commitments,'' in
  \emph{Advances in Cryptology -- {CRYPTO}~2001}, ser. Lecture Notes in
  Computer Science, J.~Kilian, Ed., vol. 2139.\hskip 1em plus 0.5em minus
  0.4em\relax Springer, Heidelberg, Aug. 2001, pp. 19--40.

\bibitem{STOC:CLOS02}
R.~Canetti, Y.~Lindell, R.~Ostrovsky, and A.~Sahai, ``Universally composable
  two-party and multi-party secure computation,'' in \emph{34th Annual {ACM}
  Symposium on Theory of Computing}.\hskip 1em plus 0.5em minus 0.4em\relax
  {ACM} Press, May 2002, pp. 494--503.

\bibitem{TCC:Garay04}
J.~A. Garay, ``Efficient and universally composable committed oblivious
  transfer and applications,'' in \emph{TCC~2004: 1st Theory of Cryptography
  Conference}, ser. Lecture Notes in Computer Science, M.~Naor, Ed., vol.
  2951.\hskip 1em plus 0.5em minus 0.4em\relax Springer, Heidelberg, Feb. 2004,
  pp. 297--316.

\bibitem{C:DamNie03}
I.~Damg{\aa}rd and J.~B. Nielsen, ``Universally composable efficient multiparty
  computation from threshold homomorphic encryption,'' in \emph{Advances in
  Cryptology -- {CRYPTO}~2003}, ser. Lecture Notes in Computer Science,
  D.~Boneh, Ed., vol. 2729.\hskip 1em plus 0.5em minus 0.4em\relax Springer,
  Heidelberg, Aug. 2003, pp. 247--264.

\bibitem{EC:Katz07}
J.~Katz, ``Universally composable multi-party computation using tamper-proof
  hardware,'' in \emph{Advances in Cryptology -- {EUROCRYPT}~2007}, ser.
  Lecture Notes in Computer Science, M.~Naor, Ed., vol. 4515.\hskip 1em plus
  0.5em minus 0.4em\relax Springer, Heidelberg, May 2007, pp. 115--128.

\bibitem{EC:CSSW06}
C.~Cr{\'e}peau, G.~Savvides, C.~Schaffner, and J.~Wullschleger,
  ``Information-theoretic conditions for two-party secure function
  evaluation,'' in \emph{Advances in Cryptology -- {EUROCRYPT}~2006}, ser.
  Lecture Notes in Computer Science, S.~Vaudenay, Ed., vol. 4004.\hskip 1em
  plus 0.5em minus 0.4em\relax Springer, Heidelberg, May~/~Jun. 2006, pp.
  538--554.

\bibitem{ICITS:CreWul08}
C.~Cr{\'e}peau and J.~Wullschleger, ``Statistical security conditions for
  two-party secure function evaluation,'' in \emph{ICITS 08: 3rd International
  Conference on Information Theoretic Security}, ser. Lecture Notes in Computer
  Science, R.~Safavi-Naini, Ed., vol. 5155.\hskip 1em plus 0.5em minus
  0.4em\relax Springer, Heidelberg, Aug. 2008, pp. 86--99.

\bibitem{JIT:DGMN13}
R.~Dowsley, J.~van~de Graaf, J.~M{\"u}ller-Quade, and A.~C.~A. Nascimento, ``On
  the composability of statistically secure bit commitments,'' \emph{Journal of
  Internet Technology}, vol.~14, no.~3, pp. 509--516, 2013.

\bibitem{EC:WolWul06}
S.~Wolf and J.~Wullschleger, ``Oblivious transfer is symmetric,'' in
  \emph{Advances in Cryptology -- {EUROCRYPT}~2006}, ser. Lecture Notes in
  Computer Science, S.~Vaudenay, Ed., vol. 4004.\hskip 1em plus 0.5em minus
  0.4em\relax Springer, Heidelberg, May~/~Jun. 2006, pp. 222--232.

\bibitem{EC:KKMPS16}
D.~Khurana, D.~Kraschewski, H.~K. Maji, M.~Prabhakaran, and A.~Sahai, ``All
  complete functionalities are reversible,'' in \emph{Advances in Cryptology --
  {EUROCRYPT}~2016, Part II}, ser. Lecture Notes in Computer Science,
  M.~Fischlin and J.-S. Coron, Eds., vol. 9666.\hskip 1em plus 0.5em minus
  0.4em\relax Springer, Heidelberg, May 2016, pp. 213--242.

\end{thebibliography}
\end{document}